\documentclass{article} 

\usepackage{amsthm}
\usepackage{amsmath}

\usepackage{mathtools}
\usepackage{bbm}

\usepackage{amssymb}
\usepackage{bbm}
\usepackage{xcolor}
\usepackage{comment}
\usepackage{tikz, tikz-3dplot}
\usepackage[T1]{fontenc}

\usetikzlibrary{arrows}
\usetikzlibrary{decorations.pathmorphing}
\pgfarrowsdeclarecombine*{spaced stealth'}{spaced stealth'}{stealth'}{stealth'}{space}{space}
\tikzset{>=spaced stealth'}

\newcommand*{\T}{^{\mkern-1.5mu\mathsf{T}}}

\usepackage{a4}
\usepackage{todonotes} 
\usepackage{utf8math}
\usepackage{enumitem} 
\usepackage{mathrsfs}

\newcommand{\bin}{\text{bin}}
\newcommand{\eps}{\varepsilon}

\newcommand{\norm}[1]{\lVert #1\rVert}
\DeclareMathOperator{\poly}{poly}

\usepackage[capitalize]{cleveref}

\newcommand{\defproblem}[3]{
	\begin{center}
		\noindent\fbox{
		\begin{minipage}{0.95\textwidth}
			#1 \\
			{\bf{Input:}} #2  \\
			{\bf{Task:}} #3
		\end{minipage}
	}
	\end{center}
}
\newtheorem{theorem}{Theorem}
\newtheorem{lemma}[theorem]{Lemma}
\newtheorem{hypothesis}[theorem]{Hypothesis}
\newtheorem{claim}[theorem]{Claim}

\title{Fine-Grained Equivalence for Problems Related to Integer Linear Programming}

\author{Lars Rohwedder\footnote{Maastricht University, Maastricht, Netherlands. Supported by Dutch Research Council (NWO) project “The Twilight
Zone of Efficiency: Optimality of Quasi-Polynomial Time Algorithms” [grant number OCEN.W.21.268]} \and Karol
W\k{e}grzycki\footnote{Saarland University and Max Planck Institute for Informatics,
        Saarbr\"ucken, Germany. 
    This work is part of the project TIPEA that has
    received funding from the European Research Council (ERC) under the European Unions Horizon
2020 research and innovation programme (grant agreement No. 850979).}
}

\date{}

\begin{document}
\maketitle
\begin{abstract}
Integer Linear Programming with $n$ binary variables and $m$ many $0/1$-constraints can be solved in time $2^{\tilde O(m^2)} \text{poly}(n)$ and it is open whether the dependence on $m$ is optimal. 
Several seemingly unrelated problems, which include variants of Closest String, Discrepancy Minimization, Set Cover, and Set Packing, can be modelled as Integer Linear Programming with $0/1$ constraints to obtain algorithms with the same running time for a natural parameter $m$ in each of the problems.
Our main result establishes through fine-grained reductions that these problems are equivalent, meaning that a $2^{O(m^{2-\varepsilon})} \text{poly}(n)$ algorithm with $\varepsilon > 0$ for one of them implies such an algorithm for all of them.

In the setting above, one can alternatively obtain an $n^{O(m)}$ time algorithm for Integer Linear Programming using a straightforward dynamic programming approach, which can be more efficient if $n$ is relatively small (e.g., subexponential in $m$).
We show that this can be improved to ${n'}^{O(m)} + O(nm)$, where $n'$ is the number of distinct (i.e., non-symmetric) variables.
This dominates both of the aforementioned running times.

\end{abstract}

\section{Introduction}
The study of parameterized complexity for Integer Linear Programming has a long history:
classical works
by Lenstra~\cite{lenstra1983integer} and Kannan~\cite{kannan1987minkowski} and very recently Rothvoss and Reis~\cite{reis2023subspace} provide FPT algorithms in the number of variables $n$ of an ILP of the form
$Ax \le b, \ x\in \mathbb Z^n$.
In an orthogonal line of research,
Papadimitriou~\cite{papadimitriou1981complexity} gave an FPT algorithm in the number of constraints and the size of the coefficients of $A$ and $b$ for an ILP in standard form $Ax = b, \ x\in \mathbb Z^n_{\ge 0}$.
Interest in the second line of work has been renewed by the improved algorithms due to Eisenbrand, Weismantel~\cite{eisenbrand2019proximity} and Jansen, Rohwedder~\cite{jansen2023integer}, which give essentially optimal running times $(m\Delta)^{O(m)} \poly(n)$ due to a conditional lower bound based
on the Exponential Time Hypothesis (ETH)~\cite{knop2020tight}. Here, $\Delta$ is the maximum absolute size of an entry in $A$.
The work by Eisenbrand and Weismantel also considers a version where variables are subject to box-constraints,
which will be the primary focus of this work, see definition below.
\defproblem{Integer Linear Programming}{
	Constraint matrix $A\in \{-\Delta,\dotsc,\Delta\}^{m\times n}$, right-hand side $b\in \mathbb Z^m$, variable bounds $\ell, u\in \mathbb Z_{\ge 0}^n$.
	}{
		Find $x\in \mathbb Z^n$ with
		\begin{align*}
			A x &= b \\
			\ell_i \le &\ x_i \le u_i &i=1,2,\dotsc,n.
		\end{align*}
	}

We refer to the variant where $\ell = (0,\dotsc,0)\T$ and $u = (1,\dotsc,1)\T$ as
\emph{binary} \textsc{Integer Linear Programming}.
\defproblem{Binary Integer Linear Programming}{
	Constraint matrix $A\in \{-\Delta,\dotsc,\Delta\}^{m\times n}$, right-hand side $b\in \mathbb Z^m$.
	}{
        Find $x\in \mathbb \{0,1\}^n$ with
		\begin{align*}
			A x &= b.
		\end{align*}
	}
The running times obtained in~\cite{eisenbrand2019proximity} for either variant is $(m\Delta)^{O(m^2)}\poly(n)$. Also for matrices with only $0/1$ coefficients nothing better than $2^{\tilde O(m^2)}\poly(n)$ is known.
It is an intriguing question whether the slightly unusual exponent of $\tilde O(m^2)$ is necessary,
which is in the spirit of \emph{fine-grained complexity}.
Since the dominant complexity-theoretic assumption of P$\neq$NP is not powerful enough
to show precise lower bounds on running times, the field of fine-grained complexity
is concerned with finding lower bounds via stronger
conjectures, see e.g., ~\cite{bringmann2019fine,vassilevska2015hardness,cygan2016problems}.
A number of such conjectures exist by now, often with interesting connections between
them. Even if one doubts these conjectures, the reductions still 
provide insights into how problems relate to each other.

Based on existing conjectures, the best lower bound known on the exponent
of \textsc{Integer Linear Programming}
is $\Omega(m\log m)$ from the easier unbounded setting~\cite{knop2020tight},
which is of course not tight in this setting.
In this paper, we take another path: we
assume that \textsc{Integer Linear Programming} cannot be solved faster than the state-of-the-art
and present several other natural parameterized
problems are all equivalent with respect to improvements on their running time.

\begin{hypothesis}[ILP Hypothesis]\label{conj:ilp}
    For every $\eps > 0$, there is no $2^{O(m^{2-\eps})} \poly(n)$ time algorithm for
    Integer Linear Programming with $\Delta = O(1)$.
\end{hypothesis}

In all of the problems below, the symbol $m$ is chosen for the parameter of interest.
Many of them are well known applications of ILP techniques, see e.g.~\cite{knop2020combinatorial}.
\defproblem{Closest String with Binary Alphabet}{
	Alphabet $\Sigma = \{0, 1\}$, strings $s_1,s_2,\dotsc,s_m\in \Sigma^n$
	}{
		Find string $t\in \Sigma^n$ minimizing
		\begin{equation*}
			\max_{i} d(t, s_i) ,
		\end{equation*}
		where $d(t, s_i)$ is the Hamming distance between $t$ and $s_i$, i.e., the number of positions the two strings differ in.
	}
	We refer to the generalization with arbitrary $\Sigma$ simply as \textsc{Closest String}.

\defproblem{Discrepancy Minimization}{
	Universe $U = \{1,2,\dotsc, n\}$, set system $S_1,S_2,\dotsc,S_m\subseteq U$
	}{
		Find coloring $\chi : U \rightarrow \{-1, 1\}$ minimizing
		\begin{equation*}
			\max_{i} |\sum_{u\in S_i} \chi(u)| .
		\end{equation*}
	}

\defproblem{Set Multi-Cover}{
	Universe $U = \{1,2,\dotsc,m\}$, set system $S_1,S_2,\dotsc,S_n\subseteq U$, $b\in \mathbb N$
	}{
		Find $I\subseteq \{1,2,\dotsc,n\}$ of minimal cardinality such that for each $v\in U$
		there are at least $b$ sets $S_i$ with $i\in I$.
	}

\defproblem{Set Multi-Packing}{
	Universe $U = \{1,2,\dotsc,m\}$, set system $S_1,S_2,\dotsc,S_n\subseteq U$, $b\in \mathbb N$
	}{
		Find $I\subseteq \{1,2,\dotsc,n\}$ of maximal cardinality such that for each $v\in U$
		there are at most $b$ sets $S_i$ with $i\in I$.
	}
	As mentioned above, our main result is the following equivalence.
	\begin{theorem}\label{thm:equivalences}
        The following statements are equivalent:
        \begin{enumerate}[label=(\arabic*)]
            \item There exists an $2^{O(m^{2-\eps})} \poly(n)$ algorithm for \textsc{Integer Linear Programming} with $\Delta = O(1)$ when
                $\eps > 0$. \label{en:ilp}
            \item There exists an $2^{O(m^{2-\eps})} \poly(n)$ algorithm for
                \textsc{Binary Integer Linear Programming} with $A\in \{0,
                1\}^{m\times n}$ and $n\le m^{O(m)}$ \label{en:bin-ilp} with
                $\eps > 0$.\label{en:bilp}
			\item There exists an $2^{O(m^{2-\eps})} \poly(n)$ algorithm for
                \textsc{Closest String with Binary Alphabet} with $\eps >
                0$.\label{en:closest}
			\item There exists an $2^{O(m^{2-\eps})} \poly(n)$ algorithm for
                \textsc{Discrepancy Minimization} with $\eps >
                0$.\label{en:disc}
			\item There exists an $2^{O(m^{2-\eps})} \poly(n)$ algorithm for
                \textsc{Set Multi-Cover} with $\eps > 0$.\label{en:cover}
			\item There exists an $2^{O(m^{2-\eps})} \poly(n)$ algorithm for
                \textsc{Set Multi-Packing} \label{en:packing} with $\eps >
                0$.\label{en:pack}
		\end{enumerate}
	\end{theorem}
	Note that \cref{en:ilp} is the negation of \Cref{conj:ilp}.
	All problems in \Cref{thm:equivalences} are easily transformed into the
    first problem, i.e., \textsc{Integer Linear Programming} with $\Delta =
    O(1)$, while maintaining the same value of $m$. Hence, the more interesting aspect of the theorem is that all these problems are as expressive as the first one.

	\Cref{conj:ilp} considers \textsc{Integer Linear Programming} with relatively small entries, i.e.,
    $\Delta = O(1)$. One can also ask the question of whether there is any parameter regime for $\Delta$ for
    which the state-of-the-art can be improved. In this spirit, a stronger variant of the conjecture is the following.
    \begin{hypothesis}[Strong ILP Hypothesis]\label{conj:s-ilp}
	    For every $\eps > 0$ and $\delta \ge 0$, there is no $2^{O(m^{\delta + 2-\eps})} \poly(n)$ time algorithm for Integer Linear Programming with $\Delta = 2^{m^{\delta}}$.
    \end{hypothesis}
    Note that~\cref{conj:ilp} is a special case of~\cref{conj:s-ilp} for $\delta = 0$.
	Another interesting regime is the complexity of \textsc{Integer Linear Programming} with $\Delta = 2^m$,
	because of a connection to block-structure integer programming, which we elaborate on later.
	There, the state-of-the-art algorithm requires time $m^{O(m^3)} \poly(n)$,
	\textsc{Integer Linear Programming} with large entries can be reduced to an equivalent instance with
	a $0/1$ matrix as seen in
	the following theorem, but the reduction is not strong enough to show equivalence between the two
	hypotheses.
	\begin{theorem}\label{thm:large-entries}
		There is a polynomial time algorithm that transforms an instance of
		\textsc{Integer Linear Programming} with $\Delta > 1$ into an equivalent one with
		$A' \in \{0,1\}^{m'\times n'}$ for $m' = O(m\log\Delta)$ and $n'\le {m'}^{O(m')}$.
	\end{theorem}
	This implies that if there is an algorithm with running time $2^{O(m^{1.5-\eps})}\poly(n)$ for \textsc{Integer Linear Programming} with $A \in \{0,1\}^{m\times n}$, then there is a $2^{O(m^{3-\eps'})}\poly(n)$ time algorithm for \textsc{Integer Linear Programming} with $\Delta = 2^m$.

	One might hope to improve the theorem to $m' = O(m \sqrt{\log \Delta})$, since then a $2^{O(m^{2-\eps})}\poly(n)$ time algorithm for $0/1$ matrices would imply a $2^{O(m^{3-\eps'})}\poly(n)$ time
	algorithm for $\Delta = 2^m$. However, such a reduction would imply the strong result
	that under ETH is equivalent to \Cref{conj:ilp}. This is because
under ETH the \textsc{Subset Sum} problem, i.e., the case when $m = 1$, cannot be
solved in $\Delta^{o(1)} \poly(n)$ time ~\cite{abboud2022seth} and the hypothetical reduction would be able to encode
an instance of \textsc{Subset Sum} into an ILP with $m = O(\sqrt{\Delta})$.
	We are not aware of any meaningful reduction in the other direction, i.e., from large $m$ and small $\Delta$ to smaller $m$ and larger $\Delta$. It is possible to aggregate $m$ constraints into a single one
	with entries bounded by $\Delta' = \Delta^{O(m^2)}$, but this reduction seems useless since the resulting
	parameter range requires $\poly(\Delta') \cdot \poly(n)$ time due to the ETH lower bound mentioned above.

    Assuming~\cref{conj:s-ilp}, we derive a tight lower for a form of block-structured integer linear programs that has
	been studied extensively in recent literature.
	For simplicity, we consider here the basic setting with $m\times m$ submatrices.

\defproblem{$n$-Fold Integer Linear Programming}{
	Square matrices $A_1,\dotsc,A_n,B_1,\dotsc,B_n\in \{-\Delta,\dotsc,\Delta\}^{m\times m}$, right-hand sides $b^{(0)},\dotsc,b^{(n)}\in \mathbb Z^m$.
	}{
		Find $x^{(1)},\dotsc,x^{(n)}\in \mathbb Z_{\ge 0}^{m}$ with
		\begin{align*}
			A_1 x^{(1)} + \dotsc + A_n x^{(n)} &= b^{(0)} \\
			B_1 x^{(1)} &= b^{(1)} \\
			\vdots \\
			B_n x^{(n)} &= b^{(n)}
		\end{align*}
	}

    \begin{theorem}\label{thm:nfold}
        For every $\delta > 0$, 
		there is no algorithm with running time $2^{O(m^{3-\delta})}\poly(n)$
        for \textsc{$n$-Fold Integer Linear Programming} when the maximum
        absolute entry is bounded by $\Delta = O(1)$, unless~\cref{conj:s-ilp} is false.
	\end{theorem}
	This matches the best algorithms known for the problem, see~\cite{cslovjecsek2021block} and references therein. The reduction follows the same idea as used in~\cite{hunkenschroder2024tight},
	where the authors show a non-tight quadratic lower bound for the exponent based on ETH.
	Our lower bound	is stronger simply because the conjecture we base it on is
    stronger.

	\subsection{Tightness of more general problems}
	There are a number of other, more general problems to the ones mentioned above,
	for which known algorithms would be tight assuming that one cannot improve the running
	time for the class of problems in \Cref{thm:equivalences}.
	However, for these, we do not know if they are all equivalent.

	The algorithm by Eisenbrand and Weismantel~\cite{eisenbrand2019proximity} also works for the
	optimization version of ILP, i.e.,
	\begin{equation*}
		\max c\T x, \ A x = b, \ \ell \le x \le u, \ x\in\mathbb Z^n .
	\end{equation*}
	This leads to the same running time of $(m\Delta)^{O(m^2)}\poly(n)$ except for a slightly higher constant in the exponent. Notably, the coefficients of $c$ do not increase the number
	of arithmetic operations.

	Given a matrix $A\in \{-\Delta,\dotsc,\Delta\}^{m\times n}$ and
	a convex function $g : \mathbb R^m \rightarrow \mathbb R\cup\{\infty\}$, Dadush,
	L\'eonard, Rohwedder, and Verschae~\cite{dadush2023optimizing} have shown that one can find the minimizer of
	$g(Ax), \ x \in \{0, 1\}^n$ in time $(m\Delta)^{O(m^2)}\poly(n)$ (assuming polynomial time to evaluate $g$).
	This problem is a generalization of \textsc{Binary Integer Linear Programming}, since one can take $g(b') = 0$ if $b' = b$ and $g(b') = \infty$ otherwise.

	Given a linear matroid $M = (E, \mathcal I)$, where $|E| = n$, 
	a matrix $A\in\{-\Delta,\dotsc,\Delta\}^{n\times m}$
	and a right-hand side $b\in\mathbb Z^m$, Eisenbrand, Rohwedder, and W\k{e}grzycki~\cite{eisenbrand2024sensitivity}
	gave an algorithm that finds a basis $B$ of $M$ such that $A x_B = b$ in time
	$O(m\Delta)^{O(m^2)}\poly(n)$, where
	$x_B$ is the incidence vector of $B$. This generalizes \textsc{Binary Integer Linear Programming}, since one can take $E$ as the set of binary variables and $n$ additional dummy variables
	and then take $M$ as the uniform matroid of rank $n$.
	
	Finally, \textsc{Closest String} (with arbitrary alphabet) can still be solved in time
	$m^{O(m^2)}\poly(n)$ for example by casting it as the previous matroid problem~\cite{eisenbrand2024sensitivity} or as a block structured ILP, see~\cite{knop2020combinatorial}.

	\subsection{Algorithm for few distinct variables}
	We show that the running time of $2^{\tilde O(m^2)} \poly(n)$ for Integer Linear Programming with $\Delta = O(1)$,
	i.e., the subject of \Cref{conj:ilp}, can be improved if the number of distinct variables is low.
	For the sake of generality, we state the algorithmic result for the optimization variant and any range of $\Delta$.
	\begin{theorem}\label{thm:ilp}
		Consider the integer programming problem
	\begin{align}
		\max\ &c\T x \notag\\
		A x &= b \label{eq:ilp} \\
		\ell_i \le x_i &\le u_i &i=1,\dotsc,n \notag\\
		x &\in \mathbb Z^n \ . \notag
	\end{align}
		Let $\Delta$ be an upper bound on the absolute value of entries of $A$.
		The problem \eqref{eq:ilp} can be solved in time
		\begin{equation*}
		n^{m+1} \cdot O(m\Delta)^m \cdot \log(\lVert u - \ell \rVert_\infty) .
		\end{equation*}
	\end{theorem}
	Using standard reductions, see e.g.~\Cref{sec:bounded-to-bin}, one may reduce $\| u - \ell \|_{\infty}$ to
	$(m\Delta)^{O(m)}$, making the logarithmic factor insignificant.

	We will now interpret this running time and point out interesting special cases.
	\paragraph*{\textsc{Binary ILP} with $A\in\{0,1\}^{m\times n}$.}
	Here, the running time above implies an ${n'}^{O(m)} + O(nm)$ time algorithm, where $n'$ is the number of distinct variables, i.e., variables that differ either in their entry of $c$ or in their column of $A$:
	one can merge identical variables in time $2^{O(m)} + O(nm)$ time, after which $\|u\|_{\infty}\le n$.
	Furthermore, without loss of generality, the rows of $A$ are linearly independent, thus $m\le n'$.
	Hence, the overall running time is
	\begin{align*}
		2^{O(m)} + O(nm) + {n'}^{m+1} \cdot O(m)^m \cdot \log n &\le {n'}^{O(m)} \cdot \log n + O(nm) \\
		&\le {n'}^{O(m)} + O(nm).
	\end{align*}
	Here, the last inequality follows from a case distinction whether $\log n$ is greater than ${n'}^{O(m)}$ or not.
	Note that in the setting without objective function we have $n'\le 2^{m}$.
	Thus, this running time is at least as good as the known $2^{\tilde O(m^2)}\poly(n)$ time algorithms.
	It also dominates the running time of $n^{O(m)}$ one would get by simple dynamic programming over the right-hand
	sides $b'$ that are feasible (for which there can be at most $n^{m}$).
	\paragraph*{\textsc{Binary ILP} with $A\in\{0,1\}^{m\times n}$ and a constant number of $1$s in each column.}
	Here, the number of distinct columns is polynomial in $m$ and the previous case implies a running time
	of $m^{O(m)} + O(nm)$, meaning that \Cref{conj:ilp} does not extend to this special case.

	\section{Reductions}
	The basic idea of all reductions for \Cref{thm:equivalences} is that we transform one problem with parameter $m$
	into another problem
	with parameter $m' = O(m \cdot \log^{O(1)} m)$. Additionally the size of the instance may only increase from $n$ to $2^{O(m^{2-\delta})}\poly(n)$ for some $\delta > 0$.
	The concrete reductions we prove can be seen in \Cref{fig:reductions}.
	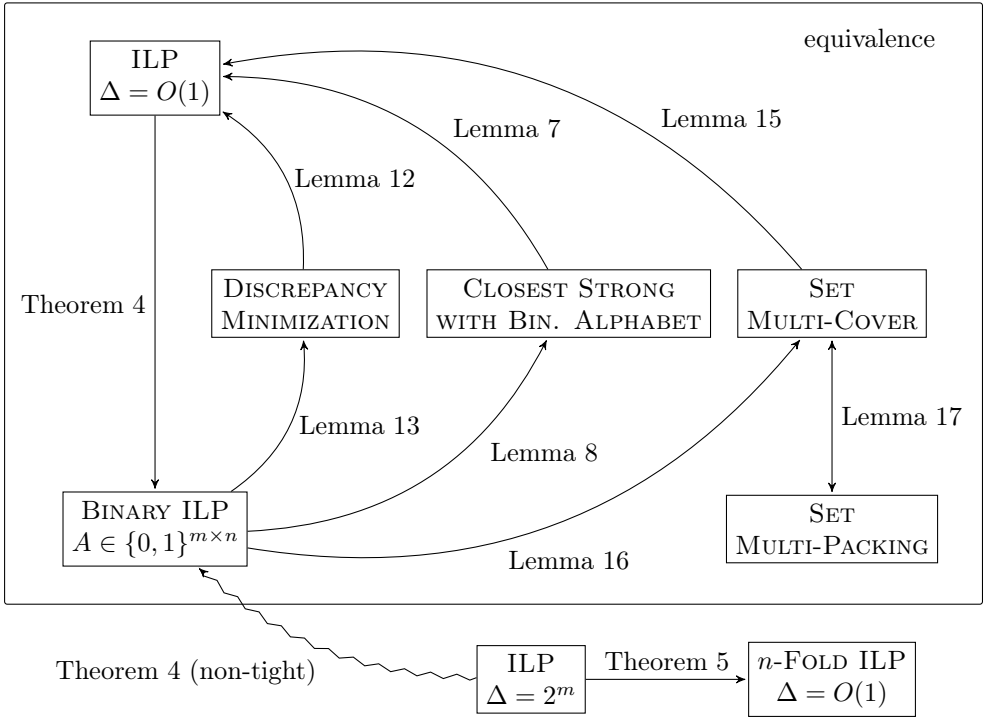
\begin{figure}
		\begin{center}
		\begin{tikzpicture}
			\node(string)[draw, align=center] at (0.5, 3) {\textsc{Closest
			Strong} \\ \textsc{with Bin. Alphabet}};
			\node(discr)[draw, align=center] at (-3, 3) {\textsc{Discrepancy} \\ \textsc{Minimization}};
			\node(cover)[draw, align=center] at (4, 3) {\textsc{Set} \\ \textsc{Multi-Cover}};
			\node(packing)[draw, align=center] at (4, 0) {\textsc{Set} \\ \textsc{Multi-Packing}};
			\node(binilp)[draw, align=center] at (-5, 0) {\textsc{Binary ILP} \\ $A\in\{0,1\}^{m\times n}$};
			\node(ilp)[draw, align=center] at (-5, 6) {\textsc{ILP} \\ $\Delta = O(1)$};
			\draw[<->] (cover) to node[pos=0.5, right] {\Cref{lem:packing-cover}} (packing);
			\draw[->] (binilp) to[bend right] node[pos=0.5, right] {\Cref{lem:discr2}} (discr);
			\draw[->] (discr) to[bend right] node[pos=0.5, right] {\Cref{lem:discr1}} (ilp);
			\draw[->] (ilp) to node[pos=0.5, left] {\Cref{thm:large-entries}} (binilp);
			\draw[->] (cover) to[bend right] node[pos=0.3, above right] {\Cref{lem:cover1}} (ilp);
			\draw[->] (binilp) to[bend right] node[pos=0.4, below right] {\Cref{lem:cover2}} (cover);
			\draw[->] (binilp) to[bend right] node[pos=0.7, below right] {\Cref{lem:string2}} (string);
			\draw[->] (string) to[bend right] node[pos=0.4, above right] {\Cref{lem:string1}} (ilp);
			\node(ilplarge)[draw, align=center] at (0, -2) {\textsc{ILP} \\ $\Delta = 2^m$};
			\node(nfold)[draw, align=center] at (4, -2) {\textsc{$n$-Fold ILP} \\ $\Delta = O(1)$};
			\draw[->] (ilplarge) to node[pos=0.5, above] {\Cref{thm:nfold}} (nfold);
			\draw[->, decoration={zigzag, amplitude=1pt}, decorate] (ilplarge) to[bend left=20] node[pos=0.5, below left] {\Cref{thm:large-entries} (non-tight)} (binilp);
			\draw[rounded corners=1pt] (-7, -1) rectangle (6, 7);
			\node at (4.5, 6.5) {equivalence};
		\end{tikzpicture}
		\end{center}
		\caption{Overview of reductions in this paper.}
		\label{fig:reductions}
	\end{figure}

	\subsection{Closest String}
	\label{sec:string}

	Let $d$ be the bound on the maximum hamming distance in the decision version
    of \textsc{Closest String with Binary Alphabet}. For a string $s\in \Sigma^n$ we denote by $s[i]$ the $i$th character. For $i \le j$ we
	write $s[i\dotsc j]$ for the substring from the $i$th to the $j$th character.
    \begin{lemma}\label{lem:string1}
	    \Cref{thm:equivalences}, Statement~\ref{en:ilp} implies Statement~\ref{en:closest}
    \end{lemma}
    \begin{proof}
	The following is an ILP model for \textsc{Closest String with Binary Alphabet}.
	\begin{align*}
		\sum_{i=1}^n x_i \cdot \mathbbm{1}_{\{s_j[i] = 0\}} + (1 - x_i) \cdot \mathbbm{1}_{\{s_j[i] = 1\}} &\le d &\text{ for all } j\in\{1,2,\dotsc,m\} \\
		x &\in \{0, 1\}^n
	\end{align*}
	One may add slack variables to turn the inequalities into equalities.
    \end{proof}
    \begin{lemma}\label{lem:string2}
        \Cref{thm:equivalences}, Statement~\ref{en:closest} implies Statement~\ref{en:bilp}
    \end{lemma}
    \begin{proof}
    We want to transform the following ILP
    \begin{align}
		A x &= b \label{eq:string-ilp} \\
		x &\in \{0, 1\}^n \notag
	\end{align}
    where $A\in \{0, 1\}^{m\times n}$, into an equivalent instance of \textsc{Closest String}.
	We first rewrite~\eqref{eq:string-ilp} into a more convenient form.
	\begin{claim}\label{cl:string-matrix}
        One can in polynomial time construct a matrix $C\in \{-1,1\}^{(2m+2) \times 2n}$ and some $c\in \mathbb Z^{2m+2}$ such that~\eqref{eq:string-ilp} is feasible if and only if there
		is a solution to
		\begin{align}
			C x &\le c \label{eq:string-ilp2} \\
			x &\in \{0, 1\}^{2n}. \notag
		\end{align}
		Furthermore, every feasible $x$ for~\eqref{eq:string-ilp2} has $x_1 + \cdots + x_{2n} = n$.
	\end{claim}
	This follows from simple transformations. We defer the proof until later and first
	show how the reduction follows from it.
	By comparing to the ILP formulation of \textsc{Closest String} we observe that~\eqref{eq:string-ilp2} corresponds to a ``non-uniform'' variant of
	    \textsc{Closest String}. It can be reformulated as: given strings $s_1,\dotsc,s_{2m+2}\in \{0,1\}^{2n}$ and bounds $d_1,\dotsc,d_{2m+2}\in\mathbb Z$, find a string $t \in \{0,1\}^{2n}$ such that for each $j = 1,\dotsc,m$ we have $d(t, s_j) \le d_j$. This follows from the ILP model given in the proof of \Cref{lem:string1}.
	    Furthermore, we have the guarantee that any solution has exactly $n$ ones.
	To transform this into a regular instance, we add two more strings $r_1, r_2$ and
	to each string $s_j$ we add $4n$ more characters, which makes a total of $6n$ characters per string.
	The strings of this instance $s'_1,\dotsc,s'_{2m+2},r_1,r_2$ are defined as
	\begin{center}
	\begin{tabular}{r c c c l}
		$r_1 = ($ & $0, \dotsc, 0$, & $0, \dotsc, 0,\ 0, \dotsc, 0$, & $0,
		\dotsc\dotsc \dotsc \dotsc ,0$ & $)$, \\
		$r_2 = ($ & $1, \dotsc, 1$, & $1, \dotsc, 1,\ 1, \dotsc, 1$, & $0,
		\dotsc\dotsc\dotsc\dotsc,0$ & $)$, \\
		$s_j' = ($ & $\underbrace{\quad s_j \quad}_{2n \text{ chars}}$, & $\underbrace{1, \dotsc, 1}_{n \text{ chars}},\ \underbrace{0, \dotsc, 0}_{n \text{ chars}}$, & \!\!\!\! $\underbrace{1, \dotsc, 1,}_{2n - d_j \text{ chars}} \underbrace{0, \dotsc, 0}_{d_j \text{ chars}}$ & $)$.
	\end{tabular}
	\end{center}
Here, we assume without loss of generality that $d_j\in \{0,\dotsc,2n\}$.
	We claim that there is a solution to this instance with maximum Hamming distance $2n$ if and only if there is a solution to the non-uniform instance.
	    \begin{claim}\label{cl:string1}
		    If there is a string $t'\in \{0,1\}^{6n}$ with distance at most $2n$ to $r_1, r_2$, and
		    $s'_j$, $j=1,2,\dotsc,2m+2$, then there is also a string $t\in \{0, 1\}^{2n}$ with
		    distance at most $d_j$ to $s_j$, $j=1,2,\dotsc,2m+2$.
	    \end{claim}
	    \begin{claim}\label{cl:string2}
		    If there is a string $t\in \{0, 1\}^{2n}$ with
		    distance at most $d_j$ to $s_j$, $j=1,2,\dotsc,2m+2$, then there is also a string 
		    $t'\in \{0,1\}^{6n}$ with distance at most $2n$ to $r_1, r_2$, and
		    $s'_j$, $j=1,2,\dotsc,2m+2$.
	    \end{claim}
	    From these claims, the lemma follows immediately.
    \end{proof}

	\begin{proof}[Proof of \Cref{cl:string-matrix}]
    We add variables $\bar x_1,\dotsc,\bar x_n$
	and force a solution to take exactly $n$ many ones
		\begin{align*}
		A x &= b \\
		x_1 + \cdots + x_n + \bar x_1 + \cdots + \bar x_n &= n \\
		x, \bar x &\in \{0, 1\}^n.
	\end{align*}
	Next, we change the equality constraints into inequalities and $A$ into a $-1,1$ matrix
	\begin{align*}
		A' x - {\bf 1} \bar x &\le b' \\
		-A' x + {\bf 1} \bar x &\le -b' \\
		x_1 + \cdots + x_n + \bar x_1 + \cdots + \bar x_n &\le n \\
		-x_1 - \cdots - x_n - \bar x_1 - \cdots - \bar x_n &\le - n \\
        x,\bar{x} &\in \{0, 1\}^n
	\end{align*}
	where $A' = 2A - {\bf 1}$ with $\bf 1$ being the all-ones $m\times n$ matrix
	and $b' = 2b - (n, \dotsc, n)\T$.
	\end{proof}

	\begin{proof}[Proof of \Cref{cl:string1}]
	Suppose there is a string $t'\in \{0,1\}^{6n}$ with distance at most $2n$ to each of the strings $s'_1,\dotsc,s'_{2m+2},r_1,r_2$.
    Because $d(r_1,r_2) = 4n$, string $t'$ must match $r_1$ and $r_2$ on characters where $r_1$ and
    $r_2$ match. More precisely, 
        $t'[i] = 0$ for $i \in \{4n,\ldots,6n\}$.
    Formally, this follows because of
	\begin{align*}
		4n &\ge d(t', r_1) + d(t', r_2) \\
		   &= d(t'[1\dotsc 4n], r_1[1\dotsc 4n]) + d(t'[1\dotsc 4n], r_2[1\dotsc 4n]) \\
		   &\quad + 2 d(t'[4n+1\dotsc 6n], (0,\dotsc,0)) \\
		   &\ge d(r_1[1\dotsc 4n], r_2[1\dotsc 4n]) + 2 d(t'[4n+1\dotsc 6n], (0,\dotsc,0)) \\
		   &\ge 4n + 2 d(t'[4n+1\dotsc 6n], (0,\dotsc,0)) .
	\end{align*}
        Let us now analyze the distance between $t := t'[1\dotsc 2n]$ and $s_j$ for some $j$.
        Since the last $2n$ characters of $t'$ are zero, we have $d(t'[4n+1\dotsc 6n], s'_j[4n+1\dotsc 6n]) = 2n - d_j$. Thus,
        \begin{align*}
            d(t, s_j) &\le d(t', s'_j) - d(t'[4n+1\dotsc 6n], s'_j[4n+1\dotsc 6n]) \\
            &\le 2n - (2n - d_j) = d_j .
        \end{align*}
        Thus, string $t$ is a solution for the non-uniform instance.
	\end{proof}

	\begin{proof}[Proof of \Cref{cl:string2}]
		Let $t \in \{0, 1\}^{2n}$ with $d(t[1\dotsc 2n], s_j) \le d_j$ for all $j$.
		We extend it to a string $t'\in \{0,1\}^{6n}$ by setting $t'[1\dotsc 2n] = t$,
		$t[2n+1\dotsc 3n] = (1,\dotsc,1)$, and
	$t[3n+1\dotsc 6n] = (0,\dotsc,0)$.
	Let us now verify that $t'$ has a distance at most $2n$ to each string.
	For $r_1, r_2$ note that $t[1\dotsc 2n]$ has exactly $n$ ones by guarantee of the
	non-uniform instance. Thus, the distance to $r_1$ and $r_2$ is exactly $2n$.
	Consider some $j\in\{1,\dotsc,2m+2\}$. Then
	\begin{equation*}
		d(s'_j, t) = d(s_j, t[1\dotsc 2n]) + 2n - d_j \le 2n . \qedhere
	\end{equation*}
    \end{proof}

	\subsection{Discrepancy Minimization}
	\label{sec:discrepancy}
    \begin{lemma}\label{lem:discr1}
        \Cref{thm:equivalences}, Statement~\ref{en:ilp} implies Statement~\ref{en:disc}
    \end{lemma}
    \begin{proof}
	Let $d$ be a bound on the objective in the decision version of \textsc{Discrepancy Minimization}.
	Let $A$ be the incidence matrix of the given set system.
	Then the colorings of discrepancy at most $d$ are exactly the solutions of
	\begin{align*}
		\begin{pmatrix} A \\ -A \end{pmatrix} y 
			&\le \begin{pmatrix} d \\ d \\ \vdots \\ d \end{pmatrix} \\
		y &\in \{-1, 1\}^n.
	\end{align*}
	This can be equivalently formulated as
	\begin{align*}
		\begin{pmatrix} 2A \\ -2A \end{pmatrix} x
			&\le \begin{pmatrix} A \\ -A \end{pmatrix}\begin{pmatrix} 1 \\ 1 \\ \vdots \\ 1 \end{pmatrix} + \begin{pmatrix} 2d \\ 2d \\ \vdots \\ 2d \end{pmatrix} \\
			x &\in \{0, 1\}^n.
	\end{align*}
	where $x_i = (1 + y_i) / 2$. One may translate the inequalities into equalities by
	introducing slack variables.
	Therefore, an algorithm for \textsc{Integer Linear Programming} can be used to solve \textsc{Discrepancy Minimization}.
    \end{proof}
    \begin{lemma}\label{lem:discr2}
        \Cref{thm:equivalences}, Statement~\ref{en:disc} implies Statement~\ref{en:bilp}.
    \end{lemma}
    \begin{proof}
        Consider an ILP of the form
	\begin{align}
		A x &= b \label{eq:discr-ilp}\\
		x &\in \{0, 1\}^n \notag
	\end{align}
	for $A\in \{0, 1\}^{m\times n}$ and $n\le m^{O(m)}$.
	We will construct an instance of \textsc{Discrepancy Minimization} which has discrepancy zero if and only if the ILP has a feasible solution.
	Towards this, we first reformulate the ILP above as
	\begin{align}
		A y &= b' \label{eq:discr-ilp2} \\
		y &\in \{-1, 1\}^n \notag
	\end{align}
	where $b' = 2b - A \ (1,\dotsc,1)\T$.
	Note that $x$ is feasible for~\eqref{eq:discr-ilp} if and only if $y = 2x - (1, \dotsc, 1)\T$ is feasible for~\eqref{eq:discr-ilp2}. Also, if $b' = (0, \dotsc, 0)\T$, then~\eqref{eq:discr-ilp2} is already equivalent to an instance of \textsc{Discrepancy Minimization} that tests
	for discrepancy zero. To handle the general case, we transform it into
	an equivalent system with right-hand size $(0,\dotsc,0)\T$.
	We first construct a gadget of elements that have the same color.
	\begin{claim}\label{cl:discr-construction}
	For any $k\in\mathbb N$ we can construct a pair of matrices $B, \bar B\in
    \{0, 1\}^{(2k-1) \times 2^k}$ such that there are exactly two solutions to
	\begin{align*}
		B z + \bar B \bar z &= \begin{pmatrix} 0 \\ \vdots \\ 0 \end{pmatrix} \\ 
			z,\bar z &\in \{-1, 1\}^{2^k}
	\end{align*}
	namely
	$z = (1,\dotsc,1)\T, \bar z = (-1,\dotsc,-1)\T$ and $z = (-1,\dotsc,-1)\T, \bar z = (1,\dotsc,1)\T$.
	\end{claim}
	We will defer the proof to the end of the section.
	Using this gadget with $k = \lceil \log_2 n \rceil = O(m\log m)$ we now
    replace each coefficient $b'_j$ in the previous system by the variables from the gadget.
	Note that~\eqref{eq:discr-ilp2} is infeasible if $\| b' \|_{\infty} > n$.
    Thus assume without loss of generality that $\| b' \|_{\infty} \le n \le
    2^k$.
	Let $C, \bar C\in \{0, 1\}^{2^k \times m}$ be defined as follows.
	The $j$th row of $C$ has $b'_j$ many ones at arbitrary positions if $b'_j \ge 0$ and is all zero otherwise; contrary, the $j$th row of $\bar C$ has $-b'_j$ many ones at arbitrary positions if $b'_j < 0$ and is all zero otherwise.

	Now consider the system
	\begin{align}
		A y + C z + \bar C \bar z &= \begin{pmatrix} 0 \\ \vdots \\ 0 \end{pmatrix} \notag \\
			B z + \bar B \bar z &= \begin{pmatrix} 0 \\ \vdots \\ 0 \end{pmatrix} \label{eq:discr-instance} \\ 
			y &\in \{-1, 1\}^n \notag\\
			z,\bar z &\in \{-1, 1\}^{2^k}. \notag
	\end{align}
	We claim that~\eqref{eq:discr-instance} has a solution if and only if there is a solution to~\eqref{eq:discr-ilp2}. Let $y, z, \bar z$ be a solution to the former. Notice that the negation of a solution is also feasible.
	Due to \cref{cl:discr-construction} we may assume without loss of generality that $z = (-1,\dotsc,-1)\T$ and $\bar z = (1, \dotsc, 1)\T$, negating the solution if necessary.
	It follows that
	\begin{equation*}
		C z + \bar C \bar z = -b' .
	\end{equation*}
	Thus, $Ay = b'$, which concludes the first direction.
	For the other direction, assume that there is a solution $y$ to~\eqref{eq:discr-ilp2}.
	We set $z = (-1,\dotsc,-1)\T$, $\bar z = (1,\dotsc,1)\T$, which by \Cref{cl:discr-construction}
	satisfies $B z + \bar B \bar z = (0, \dotsc, 0)\T$. As before we have that $C z + \bar C \bar z = -b'$. Thus, $y, z, \bar z$ is a solution to~\eqref{eq:discr-instance}.
	This establishes the equivalence of the initial ILP instance to~\eqref{eq:discr-instance},
	which corresponds to an instance of \textsc{Discrepancy Minimization} where we test for discrepancy
	zero with $m' = O(m\log m)$ sets.
    \end{proof}

	\begin{proof}[Proof of \cref{cl:discr-construction}]
	The existence of such a matrix can be proven by induction: for $k = 1$, we simply take $B = \bar B = (1)$.
    Now suppose that we already have a pair of matrices $B, \bar B \in
    \{0,1\}^{(2k-1) \times 2^k}$ as above.
    Then we set
	\begin{equation*}
		B' = \begin{pmatrix}
			B & 0 \smallskip\\
			1 \cdots 1 & 0 \cdots 0 \\
			0 \cdots 0 & 1 \cdots 1
        \end{pmatrix}
        \text{ and }
		\bar B' = \begin{pmatrix}
			\bar B & 0 \smallskip\\
			0 \cdots 0 & 1 \cdots 1 \\
			1 \cdots 1 & 0 \cdots 0
		\end{pmatrix}\in \{0,1\}^{(2k+1) \times 2 \cdot 2^k}
		.
	\end{equation*}
		It can easily be checked that choosing either $z = z' = (1, \dotsc, 1)\T, \bar z = \bar z' = (-1,\dotsc,-1)\T$ or
		$z = z' = (-1, \dotsc, -1)\T, \bar z = \bar z' = (1,\dotsc,1)\T$ 
		satisfies
		\begin{equation}
			B' \begin{pmatrix} z \\ z' \end{pmatrix} + \bar B' \begin{pmatrix} \bar z \\ \bar z' \end{pmatrix} = \begin{pmatrix} 0 \\ \vdots \\ 0 \end{pmatrix} . \label{eq:discr-claim}
		\end{equation}
		Now take any $z,z',\bar z, \bar z'\in \{-1,1\}^{2^{k+1}}$ that satisfy~\eqref{eq:discr-claim}.
		Then we have that $B z + \bar B \bar z = (0,\dotsc,0)\T$. Hence by induction hypothesis
		either $z = (1,\dotsc,1)\T, \bar z = (-1,\dotsc,-1)\T$ or $z = (-1,\dotsc,-1)\T, \bar z = (1,\dotsc,1)\T$. Assume for now the first case holds. Then because of the second-to-last rows of $B, \bar B$ we
		have
		\begin{equation*}
			2^k + \sum_{i=1}^{2^k} \bar z'_i = \sum_{i=1}^{2^k} z_i + \bar z'_i = 0 .
		\end{equation*}
		Hence, $\bar z' = (-1,\dotsc,-1)\T = \bar z$. Similarly, the last row of $B, \bar B$ implies that $z' = (1,\dotsc,1)\T = z$.
		Analogously, if $z = (-1,\dotsc,-1)\T, \bar z = (1,\dotsc,1)\T$ then $z' = z$ and $\bar z' = \bar z$.
	\end{proof}

	\subsection{Set Multi-Cover}
	\label{sec:cover}
    \begin{lemma}\label{lem:cover1}
        \Cref{thm:equivalences}, Statement~\ref{en:ilp} implies Statement~\ref{en:cover}.
    \end{lemma}
    \begin{proof}

	An instance of the decision version of \textsc{Set Multi-Cover} with a bound
    $d$ on the cardinality can be formulated as
	\begin{align*}
		x_1 + \cdots + x_n &\le d \\
		A x &\ge \begin{pmatrix} b \\ \vdots \\ b \end{pmatrix} \\
		x &\in \{0, 1\}^n
	\end{align*}
	where $A\in \{0, 1\}^{m\times n}$ is the incidence matrix of the given set system.
	This can easily be translated to the form of (\ref{en:ilp}) by introducing slack variables.
	Notice that this ILP has $m + 1$ constraints. Thus, a faster algorithm for ILP
	would imply a faster algorithm for \textsc{Set Multi-Cover}.
    \end{proof}

	In the remainder, we show that the converse is also true.
    \begin{lemma}\label{lem:cover2}
        \Cref{thm:equivalences}, Statement~\ref{en:cover} implies Statement~\ref{en:bilp}.
    \end{lemma}
    \begin{proof}

    First, we reduce to
	a ``non-uniform'' version of \textsc{Set Multi-Cover} where each element $v$ has a different
	demand $b_v$.
	Let $A\in \{0,1\}^{m\times n}$ and $b\in\mathbb Z_{\ge 0}^m$ and consider the solutions to
	\begin{align*}
		Ax &= b \\
		x &\in \{0, 1\}^n.
	\end{align*}
	First, we add $n$ additional binary variables $\bar x$ and the requirement that exactly $n$ variables are equal to one, i.e.,
	\begin{align*}
		Ax &= b \\
		x_1 + \cdots + x_n + \bar x_1 + \cdots + \bar x_n &= n \\
		x &\in \{0, 1\}^n \\
		\bar x &\in \{0, 1\}^n.
	\end{align*}
	This system is equivalent to the previous one, by setting $n - x_1 - \cdots - x_n$ arbitrary variables $\bar x_i$ to $1$. Next, we transform the equality constraints $Ax = b$ into inequalities:
	\begin{align*}
		Ax &\ge b \\
		({\bf 1} - A) x + {\bf 1} \bar x &\ge \begin{pmatrix} n \\ \vdots \\ n \end{pmatrix} - b \\
		x_1 + \cdots + x_n + \bar x_1 + \cdots + \bar x_n &= n \\
		x &\in \{0, 1\}^n \\
		y &\in \{0, 1\}^n.
	\end{align*}
	Here, $\bf 1$ denotes the $n\times n$ all-ones matrix.
	Note that the second constraint is equivalent to $Ax \le b$, since we fixed the number of ones in
	the solution. This ILP is feasible if and only if the optimum of the following 
	ILP is $n$.
	\begin{align*}
		\min &\ x_1 + \cdots + x_n + \bar x_1 + \cdots + \bar x_n \\
		Ax &\ge b \\
		({\bf 1} - A) x + {\bf 1} \bar x &\ge \begin{pmatrix} n \\ \vdots \\ n \end{pmatrix} - b \\
		x_1 + \cdots + x_n + \bar x_1 + \cdots + \bar x_n &\ge n \\
		x &\in \{0, 1\}^n \\
		\bar x &\in \{0, 1\}^n.
	\end{align*}
	This ILP corresponds to an instance of non-uniform \textsc{Set Multi-Cover} with $2m+1$ elements.

	To reduce a non-uniform instance $S_1,\dotsc,S_n$, $(b_v)_{v\in U}$, to a uniform instance of \textsc{Set Multi-Cover}
	we proceed as follows: add one new element and $n$ many new sets to the instance. The coverage requirement of each element is $n$. The new element is contained in each of the new sets and in none of the old ones. Thus, each new set has to be taken. Furthermore, we add each old element $v$ to $n - b_v$ many arbitrary new sets.
    \end{proof}

	\subsection{Set Multi-Packing}
	\label{sec:packing}
    \begin{lemma}\label{lem:packing-cover}
        \Cref{thm:equivalences}, Statements~\ref{en:cover} and~\ref{en:packing} are equivalent.
    \end{lemma}
    \begin{proof}
	Notice the following duality between \textsc{Set Multi-Cover} and \textsc{Set Multi-Packing}.
	Let $U = \{1,2,\dotsc,m\}$, $S_1,S_2,\dotsc,S_n$, and $b\in\mathbb N$ be an instance of \textsc{Set Multi-Cover}. Now consider the instance of \textsc{Set Multi-Packing} with universe $U$, set system $\bar S_1 = U\setminus S_1, \bar S_2 = U\setminus S_2, \dotsc, \bar S_n = U\setminus S_n$, and bounds $\bar b = n - b$. This is a bipartition between instances of \textsc{Set Multi-Cover} and \textsc{Set Multi-Packing}, i.e., it can be performed in both ways.

	For one pair of such instances, a solution $I$ for \textsc{Set Multi-Cover} is feasible if and only if $\bar I = \{1,2,\dotsc,n\}\setminus I$ is feasible for \textsc{Set Multi-Packing}. Thus, if the optimum of \textsc{Set Multi-Cover} is $k$, then the optimum of \textsc{Set Multi-Packing} is $n - k$.
    \end{proof}

	\subsection{Integer Linear Programming}
	\label{sec:binilp}

    In this section, we prove~\Cref{thm:large-entries}, i.e., we show how to
    reduce an ILP with large coefficients into a (larger) ILP with only $0/1$
    coefficients.  Furthermore, we show how to reduce an ILP with arbitrary
    upper and lower bounds into one with at most $(m\Delta)^{O(m)}$ binary
    variables.  Note that this implies that Statements~\ref{en:ilp}
    and~\ref{en:bilp} are equivalent and concludes the proof
    of~\cref{thm:equivalences}.

	\subsubsection{From large coefficients to zero-one}
	We will transform an ILP of the form
	\begin{align*}
		Ax &= b \\
		\ell \le x &\le u \\
		x &\in \mathbb Z^n
	\end{align*}
	where $A\in \{-\Delta,\dotsc,\Delta\}^{m\times n}$ into an equivalent
	one with $A'\in \{0, 1\}^{m'\times n'}$
	for $m' = O(m\log \Delta)$.
	Let $k = \lceil \log_2(\Delta) \rceil$.
	For the $j$th row of $A$ we introduce $4 (k + 1)$ 
	rows in $A'$, denoted by $r^+_0(j), {r'}^+_0(j), \dotsc, r^+_k(j), {r'}^+_k(j), r^-_0(j), {r'}^-_0(j), \dotsc, r^-_k(j), {r'}^-_k(j)$.
	Intuitively, the rows $r^+_i(j), {r'}^+_i(j)$ are symmetric rows that each stand for a value of $2^i$ in row $j$. Similarly, $r^-_i(j), {r'}^-_i(j)$ stand for a value of $-2^i$ in row $j$.
	The right-hand sides of the new ILP are $b_j$ for row $r^+_0(j)$ and zero for all other rows affiliated with $j$.

	For each column $A_i$ we derive a column of $A'$ as follows: for row $j$
	we consider the binary encoding of $A_{ji}$ and have the column in $A'$ use this
	binary encoding in $r^+_0(j), r^+_1(j), r^+_k(j)$ if $A_{ji} \ge 0$ and in
	$r^-_0(j), r^-_1(j), r^-_k(j)$ if $A_{ji} < 0$. All other entries of this column are zero.
	
	We now add auxiliary variables to shift the values from one power to another.
	For each row $j$ and each $i=0,\dotsc,k-1$ we add:
	\begin{itemize}[noitemsep]
		\item a variable with a one in row $r^+_i(j), {r'}^+_i(j)$ and $r^-_{i+1}(j)$ and
		\item a variable with a one in row $r^-_i(j), {r'}^-_i(j)$ and $r^+_{i+1}(j)$.
	\end{itemize}
	Furthermore, for each row $j$ and each $i = 0,\dotsc,k$ we add:
	\begin{itemize}\itemsep0pt
		\item a variable with a one in row ${r'}^+_i(j)$ and $r^-_i(j)$,
		\item a variable with a one in row ${r'}^-_i(j)$ and $r^+_i(j)$, and
		\item a variable with a one in row $r^-_i(j)$ and $r^+_i(j)$.
	\end{itemize}
	Note that each auxiliary variable does not change the total value for row $j$,
	i.e., the sum of $2^i$ times rows $r^+_i(j)$ and ${r'}^+_i(j)$ minus $2^i$ times $r^-_i(j)$ and ${r'}^-_i(j)$ over all $i$.
	Thus, any solution to the new ILP must form a solution to the original ILP.
	The converse is also true, since any value for one of the rows of $j$ can always be shifted to $r^+_0(j)$ via the auxiliary variables.

	\subsubsection{From bounded to binary variables}
	\label{sec:bounded-to-bin}
	Consider an ILP of the form
	\begin{align*}
		Ax &= b \\
		\ell \le x &\le u \\
		x &\in \mathbb Z^n
	\end{align*}
	where $A\in\{-\Delta,\dotsc,\Delta\}^{m\times n}$.
	We will first transform this into an equivalent ILP of the form
	\begin{align}
		A'x &= b' \label{eg:bin-ilp2} \\
		x &\in \{0, 1\}^{n'} \notag
	\end{align}
	where $A' \in \{-\Delta,\dotsc, \Delta\}^{m \times n'}$ for $n' \le (m\Delta)^{O(m)}$.
	Assume without loss of generality that each column of $A$ is different.
	Otherwise, we can merge identical columns together by adding the respective upper and lower
	bounds.
	This implies that $n\le (2\Delta+1)^m$.
	Next, we compute a vertex solution $x^*$ to the continuous relaxation $\{Ax = b, \ell \le x \le u, x\in\mathbb R^n\}$.
	The proximity bound by Eisenbrand and Weismantel~\cite{eisenbrand2019proximity} shows that there is an integer solution $z$ with
	$\|z - x^*\|_1 \le m (2m\Delta+1)^m$ if any integer solution exists.
	Thus, choosing $\ell'_i = \max\{\ell_i, \lceil x^*_i \rceil - (2\Delta+1)^m\}$
	and $u'_i = \min\{u_i, \lfloor x^*_i \rfloor + (2\Delta+1)^m\}$ for $i=1,2,\dotsc,n$,
	we can replace $\ell, u$ by $\ell', u'$ without affecting feasibility.
	By shifting the solution, we can reduce the lower bound to zero and
	we arrive at the following ILP that is equivalent to the original one.
	\begin{align*}
		Ax &= b - A\ell' \\
		x_i &\in \{0,1,\dotsc, u'_i - \ell'_i\} &\text{ for all } i=1,2,\dotsc,n
	\end{align*}
	Note that $u'_i - \ell'_i \le 2m (2m\Delta+1)^m$. Thus replacing each variable by $u'_i - \ell'_i$ binary variables we arrive at an ILP with $n' \le 2m (2m\Delta+1)^m \cdot n \le 2m (2m\Delta+1)^{2m}$ binary variables. 

	\subsection{$n$-Fold Integer Linear Programming}
	\label{sec:nfold}
    In this section, we prove~\cref{thm:nfold}. Consider the ILP
	\begin{align}
		Ax &= b \label{eq:nfold-start} \\
		x &\in \{0, 1\}^n \notag
	\end{align}
	with $A\in \{-2^m,\dotsc,2^m\}^{m\times n}$.
	We will show that this can be formulated as an equivalent \textsc{$n$-Fold Integer Linear Program}
	with parameter $m$ and $\Delta = O(1)$.
	The reduction follows a similar idea as one in~\cite{hunkenschroder2024tight},
	which derives a lower bound based on \textsc{Subset Sum}.
	Note that if~\eqref{eq:nfold-start} had arbitrary variables (not necessarily binary)
	then we can use the reduction of the previous section to transform it into a binary one.
    For $m \ge 3$, define
	\begin{equation*}
		B = \begin{pmatrix}
			1 & 0 &    & \cdots & 0 & 1 \\
            0 &  &    & \cdots & & 0 \\
			2 & -1 &     & \cdots & & \\
			  & 2  & -1 & \cdots & & \vdots \\
			  &    &    & \ddots & & \\
			  &    &    & 2 & -1 & 0
		  \end{pmatrix} \in \{-1,0,1,2\}^{m\times m}.
	\end{equation*}
	This matrix has the property that the system $Bx = (1,0,\dotsc,0)\T, x\in\mathbb Z_{\ge 0}^m$ has exactly two solutions, namely $x = (0,\dotsc,0,1)$ and $x = (1,2^1,2^2,\dotsc,2^{m-2},0)$.
	Our \textsc{$n$-Fold Integer Linear Program} has one block $A'_i, B'_i$ for each column $A_i$ of $A$.
	Matrix $B'_i$ is defined as $B$ above with right-hand side ${b'}^{(i)} = (1,0,\dotsc,0)\T$.
	Matrix $A'_i$ is derived from $A_i$ as follows: consider a coefficient $A_{ji}$.
	We rewrite
	\begin{equation*}
		A_{ji} = \lambda_0 \cdot 2^0 + \lambda_1 \cdot 2^1 + \cdots + \lambda_{m-2} \cdot 2^{m-2} \text{ with } \lambda\in \{-2,\dotsc,2\}^{m-1} .
	\end{equation*}
	Such a $\lambda$ exists since $|A_{ji}| \le 2^m$.
	Then we set the $j$th row of $A'_i$ as $(\lambda\T, 0)$. Let $x^{(i)}$ be the variables corresponding to block $A'_i, B'_i$. By choice of $B'_i, {b'}^{(i)}$ there are exactly two settings of $x^{(i)}$ that satisfy $B'_i x^{(i)} = {b'}^{(i)}$, namely $x^{(i)} = (0,\dotsc,0,1)$ and $x^{(i)} = (2^0,2^1,\dotsc,2^{m-2},0)$. Thus
	$A'_i x^{(i)}\in \{(0,\dotsc,0)\T, A_i\}$.
	Hence, the \textsc{$n$-Fold Integer Linear Program} with $b^{(0)} = b$ is equivalent to~\eqref{eq:nfold-start}.

\section{Algorithm for few distinct variables}~\label{sec:proof:ilp}
In this section, we prove~\cref{thm:ilp}. We assume without loss of generality that in \eqref{eq:ilp}
we have $\ell = (0, 0,\dotsc, 0)$. Note that otherwise, one may obtain an equivalent ILP with
$\ell = (0, 0,\dotsc, 0)$, $u' = u - \ell$, and $b' = b - A\ell$.
We first decompose each $u_i$ into a sum of powers of two with the properties
described in the following lemma. Such a construction is well-known in 
literature, see
e.g.~\cite{polak2021knapsack,chen2024faster,martello1990knapsack}.
\begin{lemma}[Cover]
    For every integer $k \in \mathbb N$ and $h \ge \lfloor \log{k}\rfloor$,
    there exist integers $c_0(k),\ldots,c_h(k) \in \{0,1,2\}$ such that
    \begin{displaymath}
        \left\{ \sum_{i=0}^{h} 2^i x_i : 0 \le x_i \le c_i(k) \text{ for
            all } 0 \le i \le h
        \right\} = \{0,1,\ldots,k\}
        .
    \end{displaymath}
    We call such a set $c_0(k),\ldots,c_h(k)$ a \emph{cover} of $k$. This set
    can be found in polynomial time in $h$.
\end{lemma}
\begin{proof}
    For $n \in \mathbb{N}$, let $\bin_i(n) \in \{0,1\}$ denote the $i$th bit of the binary
    representation of $n$. Let $B := 2^h - 1$ be the \emph{bottom} bits, and let $T
    = k - B$ represent the \emph{top} bits of $k$.
    Now, consider the sets
    \begin{align*}
	    S_B &:= \left\{ \sum_{i=0}^{h-1} 2^i x_i : 0\le x_i \le \bin_i(B) \text{ for all } 0\le i \le h-1 \right\} \text{ and} \\
	    S_T &:=
        \left\{ \sum_{i=0}^{h} 2^i x_i : 0\le x_i \le \bin_i(T) \text{ for all } 0\le i \le h \right\}.
    \end{align*}
    We will set $c_i(k) := \bin_i(B) + \bin_i(T) \in
    \{0,1,2\}$. Showing that these numbers are a cover of $k$ is equivalent to
    \begin{displaymath}
        S_B \oplus S_T := \{ b + t \mid b \in S_B, t \in S_T\} = \{0,1,\ldots,k\}.
    \end{displaymath}
    First, observe that $S_B = \{0,1,\ldots,2^{h} - 1\}$. Moreover, there is
    no gap of length at least $2^h$ in the set $S_T$, that is, for every nonzero $a \in S_T$,
    there exists a smaller $b \in S_T$ such that $a-b < 2^h$. This $b$ can be
    constructed by flipping an arbitrary bit of $a$ to $0$. Therefore, $S_B \oplus S_T$ consists of
    consecutive numbers. The smallest number in $S_B \oplus S_T$ is clearly $0$,
    and the largest is $B+T = k$.
    \end{proof}
We set $h = \lfloor \log_2(\lVert u \rVert_\infty) \rfloor$. Using the lemma above, for each $i$, we decompose
$u_i = c_0(u_i) \cdot 2^0 + c_1(u_i) \cdot 2^1 + \cdots + c_h(u_i) \cdot 2^{h}$ with its cover.
We can now naturally rewrite \eqref{eq:ilp} as
\begin{align*}
	\max \sum_{i=0}^h &2^i c\T z^{(i)} \\
    \sum_{i=0}^h 2^i Az^{(i)} &= b\\
	0 \le z^{(k)}_i &\le c_k(u_i) &k \in \{0,\ldots,h\} \\
    z^{(0)},\ldots,z^{(h)} &\in \mathbb Z^n \ . 
\end{align*}
Each $z^{(j)}$ produces a partial solution for some unknown right-hand side
$b^{(j)} = 2^j A z^{(j)}$. We write $b^{(\le j)} = 2^0 Az^{(0)} + 2^1 Az^{(1)} + \cdots + 2^j Az^{(j)}$.

Our approach is now to compute, for $j=0,1,\dotsc,h$ and for every potential
value of $b^{(\le j)}$, a solution $z^{(0)}, z^{(1)},\dotsc, z^{(j)}$.  To do
this, we first reduce the search space of relevant vectors $b^{(\le j)}$.

\begin{lemma}
    Consider $A \in \mathbb{Z}^{m\times n}$ and $b \in \mathbb{Z}^m$ with
    $\norm{A}_\infty \le \Delta$.
    Suppose that $y = y^{(0)} \cdot 2^0 + y^{(1)} \cdot 2^1 + \cdots +
    y^{(h)} \cdot 2^h$ satisfies $A y = b$ and $0 \le y^{(j)} \le u^{(j)}$
    with $u^{(j)} \in \{0,1,2\}^n$ for each $j \in \{0,\ldots,\ell\}$. Then for $b^{(\le j)} = 2^0 Ay^{(0)} + 2^1 Ay^{(1)} + \cdots + 2^j Ay^{(j)}$, we have
    \begin{equation*}
        b^{(\le j)} \equiv b \mod 2^{j+1}
    \end{equation*}
    and furthermore
    \begin{equation*}
        b^{(\le j)} \in \{- 2^{j+2} mn\Delta, \dotsc, 2^{j+2} mn\Delta\}^m \ .
    \end{equation*}
\end{lemma}
\begin{proof}
    The first property follows from the observation that
    $b^{(>j)} := b - b^{(\le j)}$ is a vector with each component being a multiple of $2^{j+1}$.
    Hence, $b^{(\le j)} \equiv b - b^{(>j)} \equiv b \mod 2^{j+1}$.

    For the second property, observe that $\lVert y^{(h)} \rVert_1 \le \lVert
    u^{(h)} \rVert_1 \le 2n$ for each $h$. Therefore,
    $\lVert b^{(\le j)} \rVert_\infty \le \sum_{h=0}^j 2^{h+1}n\Delta \le 2^{j+2} nm\Delta$.
\end{proof}
We say that a vector $b' \in \mathbb{Z}^m$ is \emph{relevant} for $j$ if $b' \equiv b \mod 2^{j+1}$ and $\lVert b' \rVert_\infty \le 2^{j+2} mn\Delta$.
Clearly, the following holds:

\begin{claim}\label{obs:size-rel}
    For every $j$, the number of relevant vectors for $j$ is $O(mn\Delta)^m$.
\end{claim}

We will now iteratively compute solutions for all vectors relevant for $j+1$ from the solutions for all relevant vectors for $j$ with $j \geq 0$.
Note that this will immediately establish~\cref{thm:ilp}.

\begin{lemma}\label{lem:iteration}
    Let $\mathcal{S}_j$ be a set of optimal solutions to~\eqref{eq:ilp} for all
    $b' \in \mathbb{Z}^m$ that are relevant for $j$. Then, in time $n^{m+1} \cdot
    O(m\Delta)^m$, we can compute a set $\mathcal{S}_{j+1}$ of optimal solutions for all $b'' \in \mathbb{Z}^m$ relevant for $j+1$.
\end{lemma}

\begin{proof}
Let $V$ be the set of all vectors $b' \in \mathbb Z^m$ with $b'
\equiv b \mod 2^{j+1}$ and $\lVert b' \rVert_\infty \le 2^{j+2} mn\Delta$.

We define an edge-weighted directed acyclic graph with vertices $\{s\} \cup V^{(0)}
\cup V^{(1)} \cup \dotsc \cup V^{(n)}$, where $V^{(0)},V^{(1)},\dotsc,V^{(n)}$ are
$n+1$ copies of $V$, and $s$ is a distinguished source vertex.

If $j \geq 0$, there is an edge from $s$ to every vertex $v_{b'} \in V^{(0)}$ such that
	the vector $v_{b'}$ corresponds to a relevant vector $b' \in \mathbb{Z}^m$ for $j$ for which~\eqref{eq:ilp} has a solution $x_{b'}$. This edge indicates feasibility, and the weight of the edge from $s$ to $v_{b'}$ is the value $c\T x_{b'}$ for the optimal solution $x_{b'}$.
In the base case where $j < 0$, there is exactly one edge of weight $0$ to the vertex
in $V^{(0)}$ that corresponds to the all-zero vector.

For each vertex in $V^{(i-1)}$ for $0 < i \le n$, there is an edge to the corresponding vertex in $V^{(i)}$ with weight zero. Further, if 
$c_{j+1}(u_i) \ge 1$, for every vertex corresponding to $b'$ in $V^{(i-1)}$,
there is an edge to the vertex in $V^{(i)}$ that corresponds to $b' + 2^j A_i$
(if it exists). The weight of this edge is $2^j c_i$. 

Similarly, if $c_{j+1}(u_i) \ge 2$, then for every vertex in $V^{(i-1)}$
that corresponds to a relevant $b'$, there is an edge to a vertex in $V^{(i)}$
that corresponds to $b' + 2^{j+1} A_i$ (if it exists), with a cost of $2^{j+1}c_i$.

Finally, we compute the longest path from $s$ to each vertex in $V^{(n)}$, and we
store these as the values of the solutions to all the relevant right-hand sides for $j+1$.

For the running time, observe that by~\cref{obs:size-rel}, we have $|V| \le
O(mn\Delta)^m$. Hence, the graph has $n^{m+1} \cdot O(m\Delta)^m$ vertices
and edges. Thus, the longest path problem can be solved in time $n^{m+1}
\cdot O(m\Delta)^m$. 

For correctness, consider a path in the graph from vertex $v_{b_1} \in
V^{(0)}$ to vertex $v_{b_2} \in V^{(n)}$ (corresponding to vectors $b_1$ and $b_2$
respectively). The edges of this path define a vector $z^{(j+1)}$ such that $0 \le
z^{(j+1)} \le c_{j+1}(u)$. Moreover, by construction, it holds that $2^j A
z^{(j+1)} + b_1 = b_2$. Finally, the weight of this path corresponds to the value
$2^j c\T z^{(j+1)}$.
\end{proof}

\bibliographystyle{plain}
\bibliography{references}
\end{document}